\documentclass[runningheads,a4paper,envcountsame]{llncs}

\usepackage{amssymb}
\usepackage{graphicx}

\usepackage{amsmath}
\usepackage{amsfonts}

\usepackage{mathrsfs}

\usepackage{stmaryrd}

\usepackage{enumerate}
\usepackage{xspace}

\usepackage{url}
 
\urldef{\mailpj}\path|petr.jancar@vsb.cz|

\usepackage{pj-macros}

\begin{document}

\mainmatter 

\title{Deciding structural liveness of Petri nets}

\author{Petr Jan\v{c}ar
}

\institute{Dept Comp. Sci., FEI, Techn. Univ. Ostrava
	(V\v{S}B-TUO),\\
	17. listopadu 15, 70833 Ostrava, Czech Rep.\\
\mailpj
}

\maketitle

\begin{abstract}
Place/transition Petri nets are a standard model for a class of
distributed systems whose reachability spaces might be infinite. One
of well-studied topics is the verification of safety and liveness
properties in this model; despite the extensive research effort, some
basic problems remain open, which is exemplified by the open
complexity status of the reachability problem. The liveness problems
are known to be closely related to the reachability problem, and many
structural properties of nets that are related to liveness have been studied. 

Somewhat surprisingly, the decidability status of the problem if a net is structurally live, i.e. if there is an initial marking for which it is live, has remained open, as also a recent paper (Best and Esparza, 2016) emphasizes. Here we show that the structural liveness problem for Petri nets is decidable. 

A crucial ingredient of the proof is the result by Leroux (LiCS 2013)
showing that we can compute a finite (Presburger) description of the
reachability set 
for a marked Petri net if this set is semilinear.
\end{abstract}

\section{Introduction}\label{sec:intro}

Petri nets are a standard tool for modeling and analysing a class of
distributed systems; we can name~\cite{reisig2013} as a recent
introductory monograph for this area.
A~natural part of the analysis of such systems is checking the 
safety and/or liveness properties, where the question of
deadlock-freeness is just one example.

The classical version of place/transition Petri nets 
(exemplified by Fig.~\ref{fig:expnet})
is used to model systems with potentially infinite state spaces; here 
the decidability and/or complexity questions for respective analysis problems 
are often intricate.
E.g., despite several decades of
research the complexity status 
of the basic problem of \emph{reachability} (can the system
get from one given configuration to another?) 
remains unclear;
we know  that the problem is $\EXPSPACE$-hard due to
a classical construction by 
Lipton (see, e.g.,~\cite{DBLP:conf/ac/Esparza96}) but the known upper
complexity bounds are not primitive recursive 
(we can refer to~\cite{DBLP:conf/lics/LerouxS15} 
and the references therein for further information).

The \emph{liveness} of a transition (modelling a system action)
is a related problem; its complementary problem
asks if for a given initial marking (modelling an initial system
configuration) the net enables to reach  
a marking 
in which the transition is dead, in the sense that it can be 
never performed in the future.
A~marked net $(N,M_0)$, i.e. a net $N$ with an initial marking $M_0$,
is live if all its transitions are live.

The close relationship of the problems of reachability and liveness 
has been clear since the early works by
Hack~\cite{DBLP:conf/focs/Hack74,DBLP:books/garland/Hack75}. 
Nevertheless, the situation  is different for the 
problem of \emph{structural} liveness that asks,
given a net $N$, if there is a marking
$M_0$ such that $(N,M_0)$ is live.
Though semidecidability of structural liveness is clear from the above
facts, the decidability question 
has been open so far: see, e.g., the overview~\cite{wimmel2008} and
in particular the recent paper  
\cite{DBLP:journals/ipl/BestE16} where this problem (STLP) is
discussed in the Concluding Remarks section.

Here we show the decidability of structural liveness, by showing the
semidecidability of the complementary problem.
The idea is to construct, for a given net $N$, a marked net $(N',M'_0)$ 
(partly sketched in Fig.~\ref{fig:decidpslp})
that works in two
phases (controlled by additional places): in the first phase, 
an arbitrary marking $M$ from the set $\calD$ of markings with at least
one dead transition is generated, and then 
$N$ is simulated in the reverse mode from $M$.
If $N$ is not
structurally live, then the projection of the
reachability set of $(N',M'_0)$ to the set $P$ of places of $N$ 
is the whole set $\Nat^P$; 
if $N$ is
structurally live, then there is $M\in\Nat^P$ such that 
the projection of any marking reachable from $M'_0$ differs from $M$.

In the first case (with the whole set $\Nat^P$) the reachability set of 
$(N',M'_0)$ is surely semilinear, i.e. Presburger definable. 
Due to a result by
Leroux~\cite{DBLP:conf/lics/Leroux13},
there is an algorithm that
finishes with a Presburger description of the reachability set of 
$(N',M'_0)$ when it is semilinear (while it can go forever when not). 
This yields the announced semidecidability.

The construction of the above mentioned (downward closed) set $\calD$ is
standard; the crucial ingredient of our proof is 
the mentioned result by Leroux. 
Another ingredient is the decidability of reachability;
nevertheless
it is not clear that the
reachability reduces to the structural liveness, and the complexity of
the structural liveness problem is left open for future research. 

Section~\ref{sec:basicdef} provides the necessary formal background,
and Section~\ref{sec:SLPdecid} shows the decidability result. 
In Section~\ref{sec:AddRem} a few comments are added, and 
in particular an
example of a net is given where the set of live markings is not semilinear.

\section{Basic definitions}\label{sec:basicdef}

By $\Nat$ we denote the set $\{0,1,2,\dots\}$. For a set $A$, by
$A^*$ we denote the set of finite sequences of elements of $A$, and 
$\varepsilon$ denotes the empty sequence.

\subsubsection*{Nets.}
A \emph{Petri net}, or just a \emph{net} for short, is a tuple
$N=(P,T,W)$ where 
$P$ and $T$ are two disjoint finite sets of \emph{places} and \emph{transitions},
respectively, and $W:(P\times T)\cup(T\times P)\rightarrow \Nat$ is
the \emph{weighted flow function}.
A \emph{marking} $M$ of $N$ is an element of $\Nat^P$,  a
mapping from $P$ to $\Nat$, often viewed also as a vector with $|P|$
components. 

\begin{figure}[ht]
\includegraphics[scale=0.5]{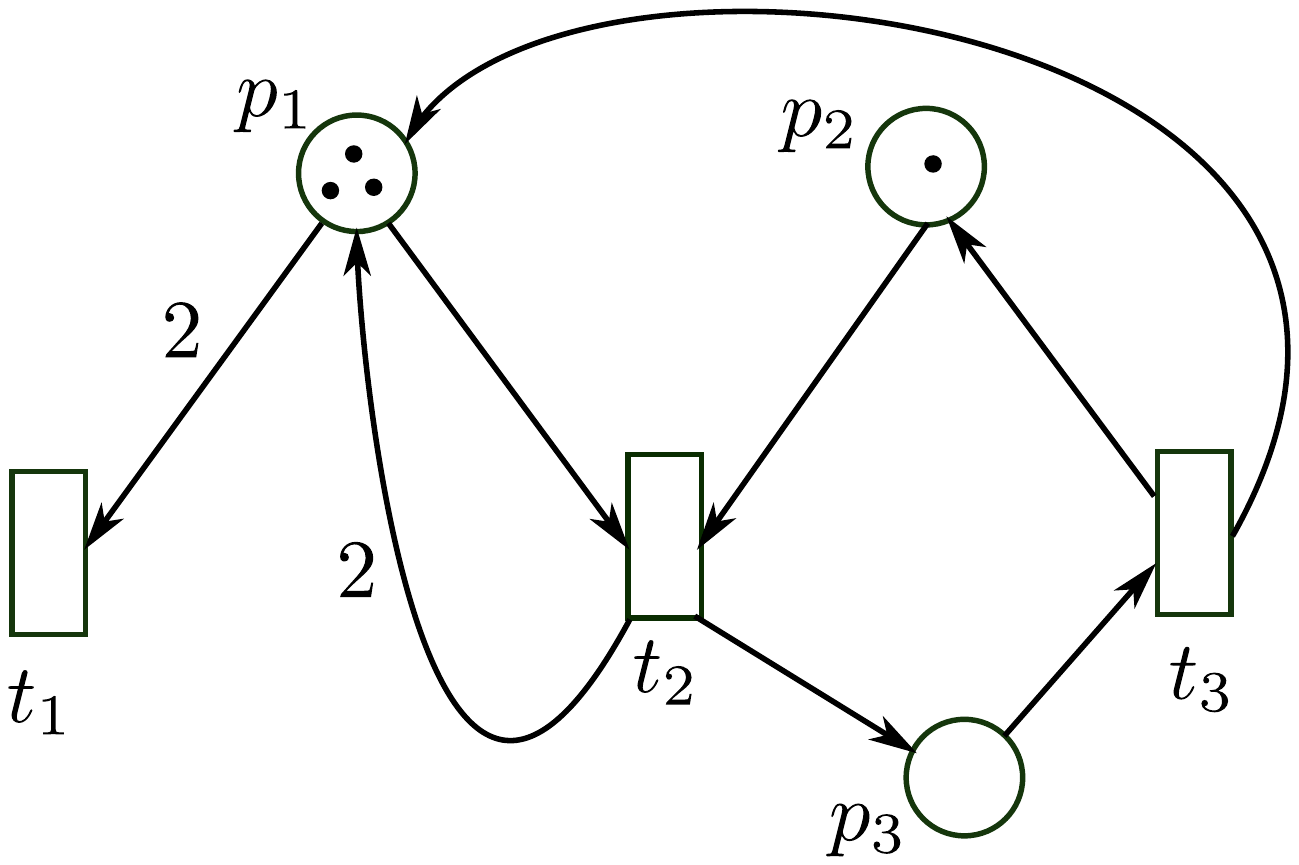}
\vspace{-30em}
\caption{Example of a net $N=(P,T,W)$, with marking $M=(3,1,0)$}\label{fig:expnet}
\end{figure}

Fig.~\ref{fig:expnet} presents a net
$N=(\{p_1,p_2,p_3\},\{t_1,t_2,t_3\},W)$ where
$W(p_1,t_1)=2$, $W(p_1,t_2)=1$, $W(p_1,t_3)=0$, etc.; 
we do not draw an arc from $x$ to $y$ when $W(x,y)=0$, and we assume 
$W(x,y)=1$ for the arcs $(x,y)$ with no depicted numbers.
Fig.~\ref{fig:expnet} also depicts a marking $M$ by using black tokens,
namely $M=(3,1,0)$,
assuming the ordering $(p_1,p_2,p_3)$ of places.

\subsubsection*{Reachability.}
Assuming a net $N=(P,T,W)$,
for each $t\in T$ we define the following relation
$\trans{t}$ on $\Nat^P$:
\begin{center}
	$M\trans{t}M'$ $\iffdef$ $\forall p\in P: M(p)\geq W(p,t)\land
	M'(p)=M(p)-W(p,t)+W(t,p)$.
\end{center}
By $M\trans{t}$ we denote that $t$ is \emph{enabled in} $M$, i.e. that
there is $M'$ such that $M\trans{t}M'$.
The relations $\trans{t}$ are inductively
extended to $\trans{u}$ for all $u\in T^*$: 
$M\trans{\varepsilon}M$;
if $M\trans{t}M'$ and $M'\trans{u}M''$, then $M\trans{tu}M''$.
The \emph{reachability set for a marking} $M$
is the set 
\begin{center}
$\rset{M}=\{M'\mid M\trans{u}M'$ for some $u\in T^*\}$.
\end{center}
For the net of Fig.~\ref{fig:expnet} we have, e.g.,
$(3,1,0)\trans{t_2}(4,0,1)\trans{t_1}(2,0,1)\trans{t_1}(0,0,1)\trans{t_3}(1,1,0)$;
we can check that
the reachability set for $(3,1,0)$ is
\begin{equation}\label{eq:rsetoddeven}
	\{\,(x,1,0)\mid x \textnormal{ is odd }\}\cup
	\{\,(y,0,1)\mid y \textnormal{ is even }\}.
\end{equation}

\subsubsection*{Liveness.} For a net $N=(P,T,W)$,
a \emph{transition} $t$ is \emph{dead in} a \emph{marking} $M$ 
if there is no $M'\in\rset{M}$ such that $M'\trans{t}$. 
(Such $t$ can be never performed in $N$ when we start from $M$.)

A \emph{transition} $t$ is \emph{live in} $M_0$ if 
there is no $M\in\rset{M_0}$ such that $t$ is dead in $M$.
(Hence for each
$M\in\rset{M_0}$ there is $M'\in\rset{M}$ such that $M'\trans{t}$.)
A~\emph{set $T'$ of transitions} is  
\emph{live in} $M_0$ if each $t\in T'$ is live in  $M_0$.
(Another natural definition of liveness of a set $T'$ is discussed in
Section~\ref{sec:AddRem}.)

A \emph{marked net} is a pair $(N,M_0)$ where  
$N=(P,T,W)$ is a net 
and $M_0$ is a marking, called \emph{the initial marking}.
A \emph{marked net} $(N,M_0)$
is \emph{live} if each transition (in other words, the set $T$)
is live in $M_0$ (in the net $N$).
A \emph{net} $N$ is \emph{structurally live} if there is $M_0$ such
that  $(N,M_0)$ is live.

E.g., the net in Fig.~\ref{fig:expnet} is structurally live since it is
live for the marking $(3,1,0)$, as can be easily checked
by inspecting the transitions enabled in the elements of the reachability set~(\ref{eq:rsetoddeven}).
We can also note that the net is not live for $(4,1,0)$,
we even have that no transition is live in  $(4,1,0)$,
since  $(4,1,0)\trans{t_1t_1}(0,1,0)$ where all transitions
are dead.

\subsubsection*{Liveness decision problems.}
\begin{itemize}
	\item
The \emph{partial liveness problem}, denoted $\plp$,
asks, given a marked net $(N,M_0)$ and a set $T'$ of its transitions,
if $T'$ is live in $M_0$.
	\item
The \emph{liveness problem}, denoted $\lp$, is a special case of
$\plp$: it
asks, given a marked net $(N,M_0)$,
if  $(N,M_0)$ is live (i.e., if all its transitions are live in
$M_0$).
	\item		
The \emph{partial structural liveness problem}, denoted $\pslp$,
asks, given a net $N$ and a set $T'$ of its transitions, if there is
$M$ in which $T'$ is live. 
\item
The \emph{structural liveness problem}, denoted $\slp$,
is a special case of $\pslp$: 
it asks, given a net $N$, if there is
$M$ such that $(N,M)$ is live. 
\end{itemize}

\section{Structural liveness of nets is decidable}\label{sec:SLPdecid}

We aim to show 
the decidability of $\pslp$, and thus also of  $\slp$:

\begin{theorem}\label{th:decidSLP}
	The partial structural liveness problem ($\pslp$) is decidable.
\end{theorem}

We prove the theorem in the rest of this section.
We first recall the famous decidability result for reachability.
The \emph{reachability problem}, denoted $\rp$, asks 
if $M\in\rset{M_0}$ when given $N,M_0,M$.

\begin{lemma}\label{lem:reachdecid}\cite{DBLP:journals/siamcomp/Mayr84}
The reachability problem ($\rp$) is decidable.
\end{lemma}	

In Petri net theory this is a fundamental theorem; we call it a
``lemma'' here, since it is one ingredient used in proving the theorem
 of this paper (i.e. Theorem~\ref{th:decidSLP}).
 The first proof of Lemma~\ref{lem:reachdecid} was given by E.~W.~Mayr 
(see~\cite{DBLP:journals/siamcomp/Mayr84} for a journal publication),
and there is a row of further papers dealing with this problem;
we can refer to a recent
paper~\cite{DBLP:conf/lics/LerouxS15} and the references therein for
further information.
The
complexity status remains far from clear;
we have $\EXPSPACE$-hardness  due to a classical construction by 
Lipton (see, e.g.,~\cite{DBLP:conf/ac/Esparza96}) but the known upper
bounds are not primitive recursive.

There are long known, and straightforward, effective reductions among the
reachability problem $\rp$ and the (partial) liveness problems ($\plp$
and $\lp$); we can find them already in Hack's works 
from 1970s~\cite{DBLP:conf/focs/Hack74,DBLP:books/garland/Hack75}.
This induces semidecidability of the partial structural liveness
problem ($\pslp$). Hence the main issue is
to establish the semidecidability of the complementary problem of
$\pslp$; roughly speaking, we need to find a finite witness when
$(N,M)$ is non-live for all $M$.

We further assume a fixed net $N=(P,T,W)$ if not said otherwise.

\subsubsection*{Sets of ``dead'' markings are downward closed.}
A natural first step for studying (partial) liveness is to explore 
the sets
\begin{center}
$\calD_{T'}=\{M\in \Nat^P\mid $ some $t\in T'$ is dead in $M\}$
\end{center}
for $T'\subseteq T$. We note that the definition entails
$\calD_{T'}=\bigcup_{t\in T'}\calD_{\{t\}}$.
E.g., in the net of Fig.~\ref{fig:expnet} we have
$\calD_{\{t_1\}}= \{(x,0,0)\mid x\leq 1\}\cup\{(0,x,0)\mid x\in\Nat\}$,
$\calD_{\{t_2,t_3\}}=\{(x,0,0)\mid x\in\Nat\}$, 
 and
\begin{equation}\label{eq:exampleDT}
\calD_T=\{(0,x,0)\mid x\in\Nat\}\cup \{(x,0,0)\mid x\in\Nat\}.
\end{equation}
Due to the monotonicity of Petri nets 
(by which we mean 
that $M\trans{u}M'$ implies $M{+}\delta\trans{u}M'{+}\delta$
for all
$\delta\in\Nat^P$), each $\calD_{T'}$ is obviously downward closed.
We say that $\calD\subseteq\Nat^P$ is \emph{downward closed} 
if $M\in \calD$ implies $M'\in \calD$ for all 
$M'\leq M$, where 
we refer to 
the component-wise order:
\begin{center}
$M'\leq M \iffdef \forall p\in P:M'(p)\leq M(p)$.
\end{center}
It is standard to characterize any downward closed subset $\calD$ of $\Nat^P$ 
by the set of its 
 maximal elements, using the extension 
 $\Natomega=\Nat\cup\{\omega\}$ where $\omega$ stands for an ``arbitrarily
 large number'' satisfying $\omega > n$ for all $n\in\Nat$.
Formally we extend a downward closed set $\calD\subseteq\Nat^P$
to the set 
\begin{center}
$\addlim{\calD}=\calD\cup\{M\in (\Natomega)^P\mid
\forall M'\in\Nat^P: M'\leq M\Rightarrow M'\in\calD\}$.
\end{center}
We thus have 
\begin{center}
$\calD=\{M'\in\Nat^P\mid M'\leq M$ for some 
$M\in \maxelem{\addlim{\calD}}\}$ 
\end{center}
where $\maxelem{\addlim{\calD}}$ is the set of maximal elements of 
$\addlim{\calD}$.
By (the standard extension of) Dickson's Lemma, the set 
$\maxelem{\addlim{D}}$ is finite.
(We can refer, e.g., to \cite{DBLP:journals/corr/abs-1208-4549} where
such completions by ``adding the limits'' are handled in a general
framework.)

E.g., for the set $\calD_T$ in~(\ref{eq:exampleDT}) we have
$\maxelem{\addlim{\calD_T}}=\{(0,\omega,0), (\omega,0,0)\}$.

\begin{proposition}\label{prop:maxconstruct}
Given $N=(P,T,W)$ and $T'\subseteq T$,
	the set $\calD_{T'}$ is downward closed and the finite set 
$\maxelem{\addlim{\calD_{T'}}}$ is effectively constructible.
\end{proposition}

\begin{proof}
The fact that $\calD_{T'}$ is downward closed has been discussed above.
A construction of the finite set $\maxelem{\addlim{\calD_{T'}}}$ 
can be easily derived once we
show that the set $\calS_{T'}=\minelem{\Nat^P\smallsetminus \calD_{T'}}$, i.e. the
set of minimal elements of the (upward closed) complement of $\calD_{T'}$,
is effectively constructible.

For each $t\in T'$,
we first compute 
$\calS_t=\minelem{\Nat^P\smallsetminus \calD_{\{t\}}}$,
 i.e. the set of minimal markings in which $t$ is not dead.
One standard possibility for computing
$\calS_t$
is to use a backward algorithm:
\begin{enumerate}
	\item		
Start with the set
$\calS_0=\minelem{\{M\mid M\trans{t}\}}$ (hence $\calS_0$ is a
singleton).
\item
For $i=0,1,2,\dots$ compute 
\begin{center}
$S_{i+1}=\minelem{\calS_i\cup \{M\mid
M\trans{t}M'\geq M''$ where $t\in T$ and $M''\in \calS_i\}}$
\end{center}
until $\calS_{i+1}=\calS_i$.
\end{enumerate}
Termination is clear by Dickson's Lemma, and $\calS_{i}=\calS_{i+1}$
obviously implies that 
$\calS_i=\calS_t$.
(Studies in a more general framework can be found, e.g., 
in~\cite{DBLP:journals/iandc/AbdullaCJT00,DBLP:journals/tcs/FinkelS01}.)

Having computed the sets 
$\calS_t=\minelem{\Nat^P\smallsetminus \calD_{\{t\}}}$ for all $t\in
T'$, we can easily compute the set 
$\calS_{T'}=\minelem{\Nat^P\smallsetminus \calD_{T'}}$ since
\begin{center}
	$\calS_{T'}=\minelem{\{M\in\Nat^P\mid (\forall t\in T')(\exists
	M'\in\calS_{t}): M\geq M'\}}$.
\end{center}
\qed
\end{proof}

\emph{Remark.}
Generally we must count with at least exponential-space algorithms for
constructing  $\maxelem{\addlim{\calD_{T'}}}$
(or $\minelem{\Nat^P\smallsetminus \calD_{T'}}$), due to Lipton's 
$\EXPSPACE$-hardness
construction that also applies to the coverability (besides
the reachability). On the other hand, 
by Rackoff's results~\cite{DBLP:journals/tcs/Rackoff78}, the numbers
in $\minelem{\Nat^P\smallsetminus \calD_{T'}}$ (and thus also the finite
numbers in $\maxelem{\addlim{\calD_{T'}}}$) are at most doubly-exponential
w.r.t. the input size, and thus fit in exponential space.
Nevertheless, the precise complexity of computing $\maxelem{\addlim{\calD_{T'}}}$
is not important in our context.

\subsubsection*{Sets of ``live'' markings are more complicated.}
Assuming $N=(P,T,W)$, for $T'\subseteq T$ we define 
\begin{center}
	$\calL_{T'}=\{M\in\Nat^P \mid T'$ is live in $M\}$. 
\end{center}
The set	$\calL_{T'}$ is not the complement of $\calD_{T'}$, but it is
obvious that $T'$ is live in $M$ iff there is no $M'$ reachable from
$M$ in which some $t\in T'$ is dead:
\begin{proposition}\label{prop:charlive}
	$M\in \calL_{T'}$ iff $\rset{M}\cap \calD_{T'}=\emptyset$.
\end{proposition}	

We note that $\calL_{T'}$ is not upward closed in general. 
We have already observed this on the net in Fig.~\ref{fig:expnet},
where $\calD_T=\{(0,x,0)\mid x\in\Nat\}\cup \{(x,0,0)\mid x\in\Nat\}$
(i.e., $\maxelem{\addlim{\calD_T}}=\{(0,\omega,0), (\omega,0,0)\}$).
It is not difficult to verify that in this net we have 
\begin{equation}\label{eq:LT}
\calL_{T}=\{\,M\in\Nat^{\{p_1,p_2,p_3\}} \mid
	M(p_2){+}M(p_3)\geq 1\textnormal{ and }
	M(p_1){+}M(p_3)
\textnormal{ is odd}\,\}. 
\end{equation}

Prop.~\ref{prop:charlive} has the following simple corollary:

\begin{proposition}\label{prop:fullnonlive}
	The answer to an instance $N=(P,T,W), T'$ of $\pslp$ (the partial
	structural liveness problem) is 
\begin{enumerate}	
\item
	YES 
	if $\calL_{T'}\neq\emptyset$, i.e., if
$\{M\in \Nat^P; \rset{M}\cap
\calD_{T'}\neq\emptyset\}\neq\Nat^P$.
\item
NO if $\calL_{T'}=\emptyset$, i.e., if
$\{M\in \Nat^P; \rset{M}\cap
\calD_{T'}\neq\emptyset\}=\Nat^P$.
	\end{enumerate}
\end{proposition}	
It turns out important for us that in the case $2$ (NO) the set 
$\{M\in \Nat^P; \rset{M}\cap
\calD_{T'}\neq\emptyset\}$ is semilinear. We now recall the relevant
facts, and then give a proof of Theorem~\ref{th:decidSLP}.

\subsubsection*{Semilinear sets.}

For a fixed (dimension) $d\in\Nat$, a \emph{set} $\linset\subseteq\Nat^d$ 
is \emph{linear}
if there is a (basic) vector $\rho\in\Nat^d$ and (period) vectors 
$\pi_1,\pi_2,\dots,\pi_k\in\Nat^d$ (for some $k\in\Nat$) such that 
\begin{center}
$\linset=\{\,\rho+x_1\pi_1+x_2\pi_2+\cdots +x_k\pi_k\mid
x_1,x_2,\dots,x_k\in\Nat\,\}$.
\end{center}
Such vectors $\rho, \pi_1,\pi_2,\dots,\pi_k$ constitute a \emph{description} of
the set $\linset$.

A \emph{set} $\semilinset\subseteq\Nat^d$ is \emph{semilinear} if
it is the union of finitely many linear sets;
a \emph{description} of $\semilinset$ is a
collection of descriptions of $\linset_i$, $i=1,2,\dots,m$  (for some
$m\in\Nat$),
where 
$\semilinset=\linset_1\cup\linset_2\cup\cdots\cup\linset_m$ 
and $\linset_i$ are linear.

It is well known that an equivalent formalism for describing semilinear sets are
Presburger formulas~\cite{GinsburgSpanier1966},
the arithmetic formulas that can use addition but no multiplication
(of variables); we also recall that the truth of (closed) Presburger formulas is
decidable.  E.g., all downward (or upward) closed sets
$\calD\subseteq\Nat^P$ are semilinear, 
and also the above sets~(\ref{eq:rsetoddeven}) and~(\ref{eq:LT}) are
examples of semilinear sets.

It is also well known that the reachability sets $\rset{M}$ 
are not semilinear in general; 
similarly the sets $\calL_{T'}$ (of live markings) are not semilinear
in general.
(We give an example in Section~\ref{sec:AddRem}.)
But we have the following result by 
Leroux~\cite{DBLP:conf/lics/Leroux13}; it is again an
important theorem in Petri net theory 
that we call a ``lemma'' in our context (since it is
an ingredient for proving Theorem~\ref{th:decidSLP}).

\begin{lemma}\label{lem:semilinconstruct}~\cite{DBLP:conf/lics/Leroux13}
There is an algorithm that, given a marked net $(N,M_0)$,
is guaranteed to halt if the reachability set $\rset{M_0}$ 
is semilinear, in which case it produces a (finite)
description of this set.
\end{lemma}

Roughly speaking, the algorithm (of Lemma~\ref{lem:semilinconstruct})
generates the reachability graph for $M_0$ while performing
``accelerations'' whenever possible (which captures repeatings of
some transition sequences by simple formulas); this process creates a sequence of
descriptions of
increasing subsets of the reachability set $\rset{M_0}$ until the
subset is closed under all steps $\trans{t}$ (which can be easily
checked); in this case the subset 
(called an inductive invariant in~\cite{DBLP:conf/lics/Leroux13})
is obviously equal 
to $\rset{M_0}$, and the process is guaranteed to reach such a case 
when $\rset{M_0}$ is semilinear. 
(A consequence highlighted in~\cite{DBLP:conf/lics/Leroux13} is that
in such a case all reachable markings can be reached by sequences of
transitions from a bounded language.)

\subsubsection*{Proof of Theorem~\ref{th:decidSLP}} (decidability of
$\pslp$).

Given $N=(P,T,W)$ and $T'\subseteq T$, we will construct a marked net 
$(N',M'_0)$ where $N'=(P\cup P_{new},T\cup T_{new}, W')$ so that 
we will have:
\begin{enumerate}[a)]
	\item		
		if $\calL_{T'}=\emptyset$ in $N$ (i.e., $T'$ is
		non-live in each marking of $N$) then $\rset{M'_0}$
		is semilinear and the restriction of $\rset{M'_0}$ to
		$P$ is $\Nat^P$;
	\item 
if $\calL_{T'}\neq\emptyset$, then there is $M\in\Nat^P$ such that
the restriction of any $M'\in\rset{M'_0}$ to $P$ is not equal to $M$.
\end{enumerate}
By this construction the proof will be finished, since in the case a)
the algorithm of Lemma~\ref{lem:semilinconstruct} applied to 
$(N',M'_0)$ is guaranteed to finish with a description of
$\rset{M'_0}$ from which it will be clear if the restriction 
of $\rset{M'_0}$ to $P$ is
$\Nat^P$; in the case b) another algorithm will find some $M\in\Nat^P$
for which the respective condition can be verified due to
a standard extension of the decidability of reachability 
(Lemma~\ref{lem:reachdecid}), as our
construction will also make quite transparent.

The constructed $(N',M'_0)$ captures the set
$\{M\in \Nat^P; \rset{M}\cap
\calD_{T'}\neq\emptyset\}$ that is highlighted in
Prop.~\ref{prop:fullnonlive}. The idea is illustrated in
Fig.~\ref{fig:decidpslp}; we create a marked net that first generates
an element of $\calD_{T'}$ and then simulates $N$ in the reverse mode. 

\begin{figure}[ht]
	\begin{center}
	\includegraphics[scale=0.5]{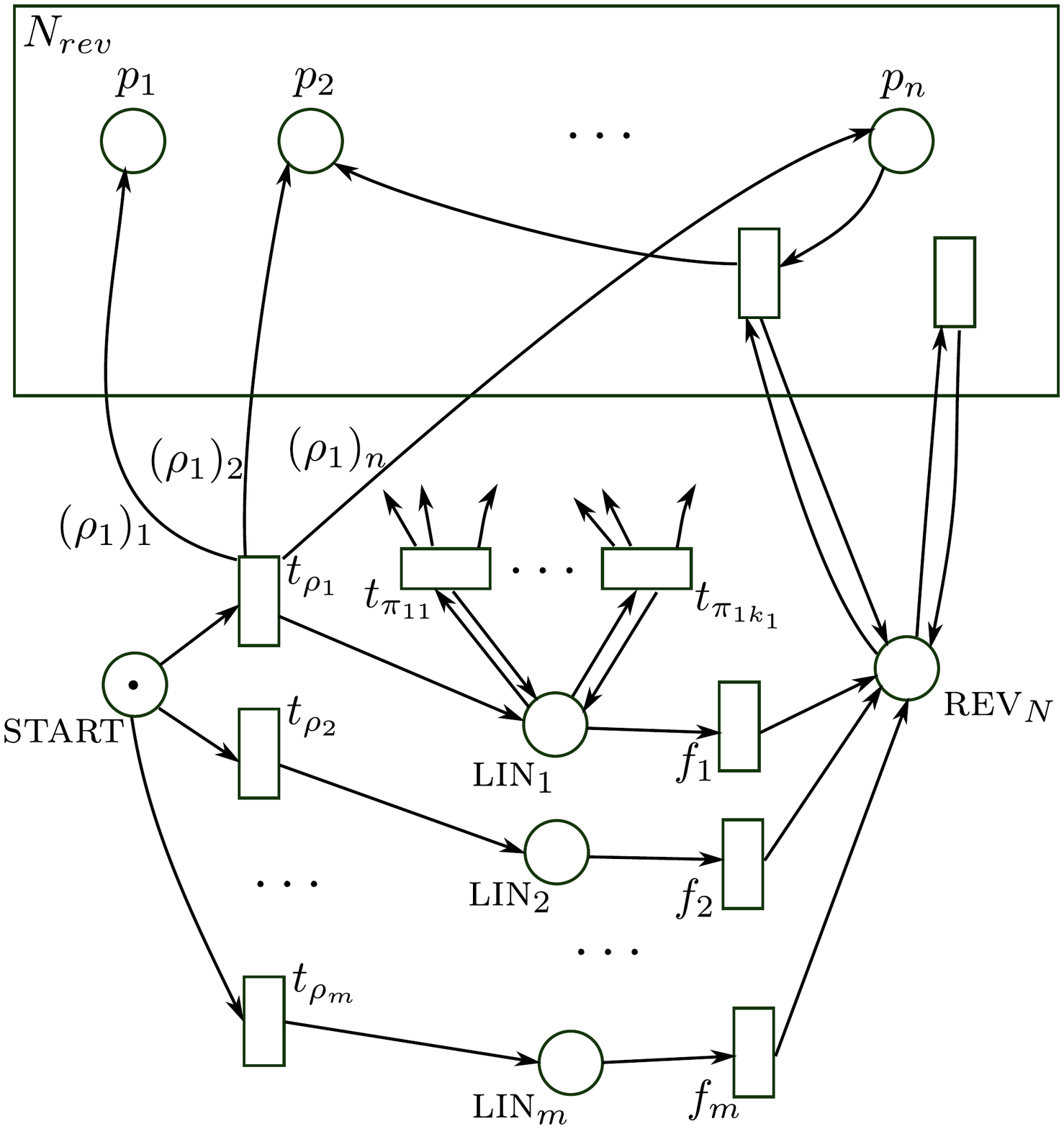}
\vspace{-12em}
\caption{Construction of $(N',M'_0)$ for deciding the (partial) structural
liveness ($\pslp$)}\label{fig:decidpslp}
\end{center}
\end{figure}

More concretely, we assume the ordering $(p_1,p_2,\dots,p_n)$ of the
set $P$ of places in $N$, and compute 
a description of the set $\calD_{T'}$ (recall
Prop.~\ref{prop:maxconstruct}); let
\begin{center}
$\calD_{T'}=\linset_1\cup\linset_2\cup\cdots\cup\linset_m$, 
\end{center}
given by
descriptions $\rho_i,\pi_{i1},\pi_{i2},\dots,\pi_{ik_i}$ of the linear
sets $\linset_i$, for $i=1,2,\dots,m$. 
(We choose this description of $\calD_{T'}$ to make clear that the
construction can be applied to any semilinear set, not only to a
downward closed one.)

The construction of $(N',M'_0)$, where 
$N'=(P\cup P_{new},T\cup T_{new}, W')$, can be now described as follows:
\begin{enumerate}
\item
Given $N=(P,T,W)$, create the ``reversed'' net  $N_{rev}=(P,T,W_{rev})$,
where $W_{rev}(p,t)=W(t,p)$ and  $W_{rev}(t,p)=W(p,t)$ for all $p\in
P$ and $t\in T$.
\\
(By induction on the length of $u$ it is easy to verify 
that $M\trans{u}M'$ in $N$ iff $M'\trans{u_{rev}}M$ in $N_{rev}$,
where $u_{rev}$ is the reversal of $u$.)
\item
To get $N'$, extend $N_{rev}$ as described below; hence 
$W'(p,t)=W_{rev}(p,t)$ and $W'(t,p)=W_{rev}(t,p)$ for all 
$p\in P$ and $t\in T$.
\item
Create the set  $P_{new}$ consisting of the new places
 $\textsc{start}$, $\textsc{lin}_1$, 
$\textsc{lin}_2$, $\dots$, $\textsc{lin}_m$, and $\textsc{rev}_N$,
and the set $T_{new}$ consisting of the 
new transitions $t_{\rho_i}$, $f_i$, and $t_{\pi_{i1}}$, $t_{\pi_{i2}}$,
$\dots$, $t_{\pi_{ik_i}}$ for all $i\in\{1,\dots,m\}$.  
(This is partly depicted in Fig.~\ref{fig:decidpslp}.)
\item
	Put $M'_0(\textsc{start})=1$ and $M'_0(p)=0$ for all
other places $p\in P\cup P_{new}$. 
\item
	For each $i\in\{1,\dots,m\}$, put $W'(\textsc{start},t_{\rho_i})=
W'(t_{\rho_i},\textsc{lin}_i)=1$, and 
$W'(t_{\rho_i},p_j)=(\rho_i)_j$ for all $j\in\{1,\dots,n\}$, where 
$(\rho_i)_j$ is the $j$-th component of the vector $\rho_i\in\Nat^n$.
(We tacitly assume that the value of $W'$ is $0$ for  
the pairs $(p,t)$ and $(t,p)$ that are not mentioned.)
\item
For each  
$t_{\pi_{i\ell}}$ ($i\in\{1,\dots,m\}$, $\ell\in\{1,\dots,k_i\}$)
put  $W'(\textsc{lin}_i,t_{\pi_{i\ell}})=
W'(t_{\pi_{i\ell}},\textsc{lin}_i)=1$,	
and $W'(t_{\pi_{i\ell}},p_j)=(\pi_{i\ell})_j$ for all
$j\in\{1,\dots,n\}$.
\item
	For each $f_i$ put
	$W'(\textsc{lin}_i,f_i)=W'(f_i,\textsc{rev}_N)=1$.
\item
For each transition $t\in T$ in $N_{rev}$ put
	$W'(\textsc{rev}_N,t)=W'(t,\textsc{rev}_N)=1$.
\end{enumerate}		
For the resulting $(N',M'_0)$, we obviously have that there is always
precisely one token on $P_{new}$; i.e., the set $\rset{M'_0}$ can be
expressed as the union
\begin{center}
$\rset{M'_0}=\calS_{\textsc{start}}
\cup\calS_{\textsc{lin}_1}\cup\cdots\cup\calS_{\textsc{lin}_m}\cup
\calS_{\textsc{rev}_N}$
\end{center}
of the disjoint sets $\calS_p=\{M\mid M\in\rset{M'_0}$ and $M(p)=1\}$ (for
$p\in\{\textsc{start},\textsc{lin}_1,\dots,\textsc{lin}_m,\textsc{rev}_N\}$).
The sets $\calS_{\textsc{start}}$,
$\calS_{\textsc{lin}_1}$, $\dots$, $\calS_{\textsc{lin}_m}$ are
obviously semilinear, and it is also clear that 
the restriction of $\calS_{\textsc{rev}_N}$ to $P=\{p_1,p_2,\dots,p_n\}$ is the set
$\{M\in \Nat^P; \rset{M}\cap
\calD_{T'}\neq\emptyset\}$ where  $\rset{M}$ refers to the net $N$. 

Hence the constructed $(N',M'_0)$ indeed satisfies the above conditions a)
and b)
(since $\calL_{T'}=\emptyset$ iff
$\{M\in \Nat^P; \rset{M}\cap
\calD_{T'}\neq\emptyset\}=\Nat^P$). To verify b), it suffices to find
a marking $M$ of $N'$ that satisfies 
$M(\textsc{rev}_N)=1$,
$M(\textsc{start})=M(\textsc{lin}_1)=\cdots=M(\textsc{lin}_m)=0$
and that is not reachable from $M'_0$.

\medskip

\emph{Remark.}
For establishing the non-reachability of $M$ from $M'_0$ we can use an
algorithm guaranteed by 
the decidability of reachability (Lemma~\ref{lem:reachdecid}).
Another option, due to another result of
Leroux (see, e.g.,~\cite{DBLP:conf/birthday/Leroux12}), is to find
a description of a semilinear set that contains $M'_0$, does not
contain $M$, and is
closed w.r.t. all steps $\trans{t}$ (being thus an inductive invariant
in the terminology of~\cite{DBLP:conf/birthday/Leroux12}).

\section{Additional remarks}\label{sec:AddRem}

\subsubsection*{Sets of live markings can be nonsemilinear.}
In Petri net theory, there are many results that relate liveness 
to specific structural properties of nets. We
can name~\cite{DBLP:conf/apn/BarkaouiP96} as 
an example of a cited paper from this area.
Nevertheless, the general
structural liveness problem is still 
not fully understood; one reason might be the fact that 
\begin{center}
the set of live markings of a given net is not
semilinear in general. 
\end{center}

We give an example.
If the set $\calL_{T}$ of live markings for the net $N=(P,T,W)$ in
Fig.~\ref{fig:nonsemilin}  was semilinear, then also its intersection 
with the set
$\{(x_1,0,1,0,1,x_6)\mid x_1,x_6\in\Nat\}$ would be semilinear (i.e.,
definable by a Presburger formula).
But is is straightforward to verify that the markings in this set are
live if, and only if, $x_6> 2^{x_1}$: any marking $M$ where $p_4$ is
marked (forever), i.e. $M(p_4)\geq 1$, is clearly live, and
 we can get at most $2^{x_1}$ tokens in
$p_5$ as long as $p_4$ is unmarked; if  $x_6\leq 2^{x_1}$, then there
is a reachable 
 marking where all transitions are dead, but otherwise $p_4$ gets
 necessarily marked.

\begin{figure}[ht]
	\begin{center}
	\includegraphics[scale=0.5]{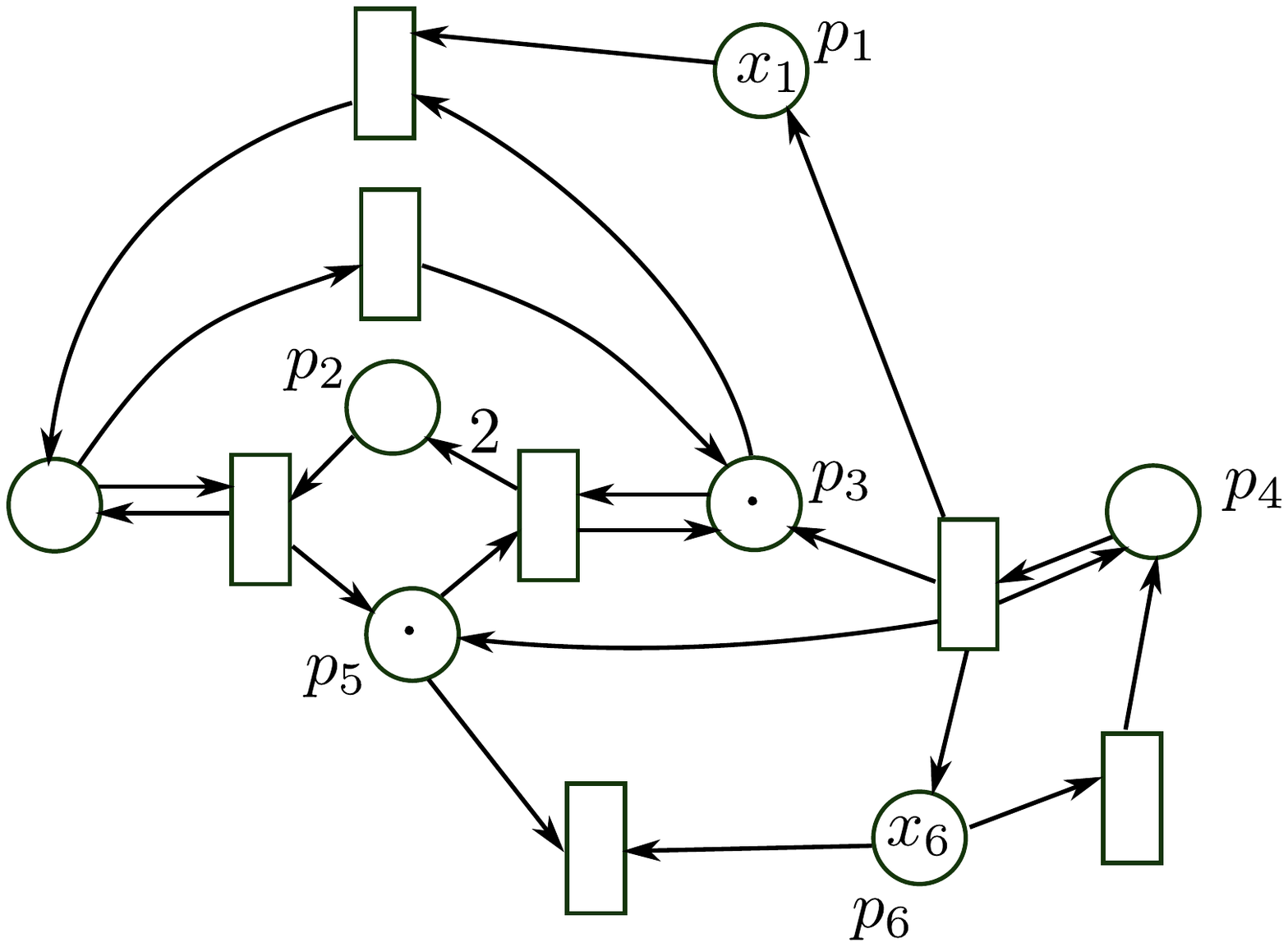}
\vspace{-26em}
\caption{Sets of live markings can be
nonsemilinear}\label{fig:nonsemilin}
\end{center}
\end{figure}

\subsubsection*{Another version of liveness of a set of transitions.}
We have defined a set $T'$ of transitions as live in a marking $M$ if
each $t\in T'$ is live in $M$. Another option is to view
$T'$ as live in $M$ if in each $M'\in\rset{M}$ at least one $t\in T'$
is not dead. But the problem if $T'$ is live in $M$ in this sense
 can be easily reduced to the problem if a specific transition is
 live, and this nuances thus make no substantial difference in our context.

\subsubsection*{Open complexity status.}
We note that it remains to be clarified what we can say about
the complexity of the (partial) structural liveness problem. 
The complexity of the (partial) liveness problem is ``close'' to the
complexity of the reachability (as follows already by the
constructions in~\cite{DBLP:conf/focs/Hack74}), but it seems
natural to expect that the \emph{structural} liveness problem
might be easier. (E.g., the boundedness problem, asking if $\rset{M_0}$
is finite when given $(N,M_0)$, is $\EXPSPACE$-complete, by the results
of Lipton and Rackoff, but the structural boundedness problem
is polynomial; here we ask, given $N$, if $(N,M_0)$ is bounded for all
$M_0$, or in the complementary way, if $(N,M_0)$ is unbounded for some
$M_0$.)

\subsection*{Acknowledgement.}
I would like to thank to Eike Best for 
drawing my attention to
the problem of structural liveness studied in this paper.

\bibliographystyle{splncs03}
\bibliography{sofsem17}

\end{document}

ZALOHA

The following lemma is standard; we just sketch the used ideas.

\begin{lemma}
The reachability problem, the partial liveness problem, and the
liveness problem are effectively inter-reducible. 
\end{lemma}

\begin{proof}
	\begin{itemize}
		\item
			LP is a special case of PLP.
		\item	
	PLP can be reduced to RP: 

	Given  $N,T',M$, we construct $\calD_{T'}$, represented by
$\maxelem{\calD_{T'}}$, and ask if there is $M'\in\rset{M}$ such that
$M'\in \calD_{T'}$. (It is easy to formulate this as an RP instance.)
\item
RP can be reduced to LP: 

RP can be reduced to SPZR (single-place zero-marking reachability).
Let us have an instance $N,M_0,p$ of SPZR.
Add a run-place $r_1$ with one token; it can be moved by a special
transition $t_1$ to $r_2$. If $p$ is nonempty at that moment, a
transition $t_2$ is enabled that puts a token in $r_3$ which is a
run-place for a transition $t_3$ that adds token to all 
places (hence once a token in $r_3$, the net is trivially live).
So if $p$ always has a token in $(N,M_0)$ then the resulting
$(N',M'_0)$ is live; if not, then it is possible to reach a marking
where all transitions are dead.
\end{itemize}
\qed
\end{proof}